\documentclass[11pt]{article}

\usepackage{times}
\usepackage{amsfonts}
\usepackage{graphicx}
\usepackage{url,hyperref}
\usepackage{algonjk,tikz,latexsym,amsmath,amsthm}

\advance \topmargin by -1.0in
\advance \oddsidemargin by -0.8in
\newcommand{\myparskip}{3pt}
\parskip \myparskip

\newtheorem{lemma}{Lemma}[section]
\newtheorem{theorem}[lemma]{Theorem}

\newtheorem{corollary}[lemma]{Corollary}
\newtheorem{proposition}[lemma]{Proposition}

\newtheorem{remark}[lemma]{Remark}

\newenvironment{proofof}[1]{\smallskip\noindent{\bf Proof of #1:}}%
        {\hspace*{\fill}$\Box$\par}

\newcommand{\etal}{{\em et al.}\ }
\newcommand{\assign}{\leftarrow}
\newcommand{\eps}{\epsilon}

\newcommand{\kec}[1]{$k$-$#1${\sc EC} }
\newcommand{\kvc}[1]{$k$-$#1${\sc VC} }
\newcommand{\ke}{\kec{2}}
\newcommand{\kv}{\kvc{2}}
\newcommand{\densV}{dens-$2${\sc VC} }
\newcommand{\densLP}{{\bf LP-dens} }

\newcommand{\ceil}[1]{\lceil #1 \rceil}

\newcommand{\opt}{\textrm{\sc OPT}}

\newcommand{\dens}[1]{\emph{density}(#1)}
\DeclareMathAlphabet{\mathpzc}{OT1}{pzc}{m}{it}
\newcommand{\script}[1]{\mathcal{#1}}

\begin{document}

\title{Min-Cost $2$-Connected Subgraphs With $k$ Terminals}
\author{
  Chandra Chekuri\thanks{Dept. of Computer Science, University of Illinois, Urbana,
    IL 61801. Partially supported by NSF grants CCF 07-28782 and CNS-0721899, and a
    US-Israeli BSF grant 2002276. {\tt chekuri@cs.uiuc.edu}}
  \and 
  Nitish Korula\thanks{Dept. of Computer Science, University of Illinois, Urbana,
    IL 61801.  Partially supported by NSF grant CCF 07-28782. {\tt
      nkorula2@uiuc.edu}} }
\date{}
\maketitle

\begin{abstract}
  In the \kv problem, we are given an undirected graph $G$ with edge
  costs and an integer $k$; the goal is to find a minimum-cost
  2-vertex-connected subgraph of $G$ containing at least $k$
  vertices. A slightly more general version is obtained if the
  input also specifies a subset $S \subseteq V$ of {\em terminals} and
  the goal is to find a subgraph containing at least $k$ terminals.
  Closely related to the \kv problem, and in fact a special case of
  it, is the \ke problem, in which the goal is to find a minimum-cost
  2-edge-connected subgraph containing $k$ vertices. The \ke problem
  was introduced by Lau \etal \cite{LauNSS07}, who also gave a
  poly-logarithmic approximation for it. No previous approximation
  algorithm was known for the more general \kv problem.  We describe
  an $O(\log n \cdot \log k)$ approximation for the \kv problem.
\end{abstract}

\section{Introduction}
\label{sec:intro}

Connectivity and network design problems play an important role in
combinatorial optimization and algorithms both for their theoretical
appeal and their many real-world applications. An interesting and
large class of problems are of the following type: given a graph
$G(V,E)$ with edge or node costs, find a minimum-cost subgraph $H$ of
$G$ that satisfies certain connectivity properties. For example, given
an integer $\lambda > 0$, one can ask for the minimum-cost spanning
subgraph that is $\lambda$-edge or $\lambda$-vertex connected. If
$\lambda=1$ then this is the classical minimum spanning tree (MST)
problem. For $\lambda > 1$ the problem is NP-hard and also APX-hard to
approximate. More general versions of connectivity problems are
obtained if one seeks a subgraph in which a subset of the nodes $S
\subseteq V$ referred to as {\em terminals} are $\lambda$-connected.
The well-known Steiner tree problem is to find a minimum-cost subgraph
that ($1$-)connects a given set $S$. Many of these problems are
special cases of the survivable network design problem (SNDP). In
SNDP, each pair of nodes $u,v \in V$ specifies a connectivity
requirement $r(u,v)$ and the goal is to find a minimum-cost subgraph
that has $r(u,v)$ disjoint paths for each pair $u,v$. Given the
intractability of these connectivity problems, there has been a large
amount of work on approximation algorithms. A number of elegant and
powerful techniques and results have been developed over the years
(see \cite{Hochbaum96,Vazirani01}). In particular, the primal-dual
method \cite{AgrawalKR95,GoemansW96} and iterated rounding \cite{Jain}
have led to some remarkable results including a $2$-approximation for
edge-connectivity SNDP \cite{Jain}.

An interesting class of problems, related to some of the connectivity
problems described above, is obtained by requiring that only $k$ of
the given terminals be connected. These problems are partly motivated
by applications in which one seeks to maximize profit given a upper
bound (budget) on the cost.  For example, a useful problem in vehicle
routing applications is to find a path that maximizes the number of
vertices in it subject to a budget $B$ on the length of the path. In
the exact optimization setting, the profit maximization problem is
equivalent to the problem of minimizing the cost/length of a path
subject to the constraint that at least $k$ vertices are included. Of
course the two versions need not be approximation equivalent,
nevertheless, understanding one is often fruitful or necessary to
understand the other. The most well-studied of these problems is the
$k$-MST problem; the goal here is to find a minimum-cost subgraph of
the given graph $G$ that contains at least $k$ vertices (or
terminals). This problem has attracted considerable attention in the
approximation algorithms literature and its study has led to several
new algorithmic ideas and applications
\cite{AwerbuchABV95,Garg96,Garg05,ChaudhuriGRT03,orienteering}.  We
note that the Steiner tree problem can be relatively easily reduced in
an approximation preserving fashion to the $k$-MST problem.  More
recently, Lau \etal \cite{LauNSS07} considered the natural
generalization of $k$-MST to higher connectivity. In particular they
defined the $(k,\lambda)$-subgraph problem to be the following: find a
minimum-cost subgraph of the given graph $G$ that contains at least
$k$ vertices and is $\lambda$-edge connected. We use the notation
\kec{\lambda} to refer to this problem. In \cite{LauNSS07} an
$O(\log^3 k)$ approximation was claimed for the \ke problem.  However,
the algorithm and proof in \cite{LauNSS07} are incorrect.  More
recently, and in independent work from ours, the authors of
\cite{LauNSS07} obtained a different algorithm for \ke that yields an
$O(\log n \log k)$ approximation. We give later a more detailed comparison
between their approach and ours. It is also shown in \cite{LauNSS07}
that a good approximation for \kec{\lambda} when $\lambda$ is large
would yield an improved algorithm for the $k$-densest subgraph problem
\cite{FeigeKP01}; in this problem one seeks a $k$-vertex subgraph of a
given graph $G$ that has the maximum number of edges. The $k$-densest
subgraph problem admits an $O(n^{\delta})$ approximation for some
fixed constant $\delta < 1/3$ \cite{FeigeKP01}, but has resisted
attempts at an improved approximation for a number of years now.

In this paper we consider the vertex-connectivity generalization of
the $k$-MST problem. We define the \kvc{\lambda} problem as follows:
Given an integer $k$ and a graph $G$ with edge costs, find the
minimum-cost $\lambda$-vertex-connected subgraph of $G$ that contains
at least $k$ vertices. We also consider the {\em terminal} version of
the problem where the subgraph has to contain $k$ terminals from a
given terminal set $S \subseteq V$. It can be easily shown that the
\kec{\lambda} problem reduces to the \kvc{\lambda} problem for any $k
\ge 1$. We also observe that the \kec{\lambda} problem with terminals
can be easily reduced, as follows, to the uniform problem where every
vertex is a terminal: For each terminal $v \in S$, create $n$ dummy
vertices $v_1, v_2, \ldots, v_n$ and attach $v_i$ to $v$ with
$\lambda$ parallel edges of zero cost. Now set $k'= kn$ in the new
graph. One can avoid using parallel edges by creating a clique on
$v_1, v_2, \ldots, v_n$ using zero-cost edges and connecting $\lambda$
of these vertices to $v$. Note, however, that this reduction only
works for edge-connectivity. We are not aware of a reduction that
reduces the \kvc{\lambda} problem with a given set of terminals to the
\kvc{\lambda} problem, even when $\lambda=2$. In this paper we
consider the \kv problem; our main result is the following.

\begin{theorem}
  \label{thm:kv}
  There is an $O(\log \ell \cdot \log k)$ approximation for the \kv
  problem where $\ell$ is the number of terminals.
\end{theorem}

\begin{corollary}
  \label{cor:ke}
  There is an $O(\log \ell \cdot \log k)$ approximation for the \ke
  problem where $\ell$ is the number of terminals.
\end{corollary}

One of the technical ingredients that we develop is the theorem below
which may be of independent interest. Given a graph $G$ with edge costs
and weights on terminals $S \subseteq V$, we define $density(H)$ for a
subgraph $H$ to be the ratio of the cost of edges in $H$ to the total
weight of terminals in $H$.

\begin{theorem}
  \label{thm:cycle}
  Let $G$ be an $2$-vertex-connected graph with edge costs and let $S
  \subseteq V$ be a set of terminals. Then, there is a simple cycle $C$
  containing at least $2$ terminals (a non-trivial cycle) such that
  the density of $C$ is at most the density of $G$. Moreover, such a
  cycle can be found in polynomial time.
\end{theorem}

Using the above theorem and an LP approach we obtain the following. 
\begin{corollary}
  \label{cor:cycle}
  Given a graph $G(V,E)$ with edge costs and $\ell$ terminals
  $S\subseteq V$, there is an $O(\log \ell)$ approximation for the
  problem of finding a minimum-density non-trivial cycle.
\end{corollary}

Note that Theorem~\ref{thm:cycle} and Corollary~\ref{cor:cycle} are of
interest because we seek a cycle with at least {\em two} terminals.  A
minimum-density cycle containing only one terminal can be found by
using the well-known min-mean cycle algorithm in directed graphs
\cite{networkflows_book}.  We remark, however, that although we
suspect that the problem of finding a minimum-density non-trivial
cycle is NP-hard, we currently do not have a proof.
Theorem~\ref{thm:cycle} shows that the problem is equivalent to the
\densV problem, defined in the next section.

\smallskip
\noindent {\bf Remark:} The reader may wonder whether \ke or \kv admit
a constant factor approximation, since the $k$-MST problem admits one.
We note that the main technical tool which underlies $O(1)$
approximations for $k$-MST problem \cite{BlumRV96,Garg96,ChudakRW} is
a special property that holds for a LP relaxation of the
prize-collection Steiner tree problem \cite{GoemansW96} which is a
Lagrangian relaxation of the Steiner tree problem. Such a property is
not known to hold for generalizations of $k$-MST including \ke and \kv
and the $k$-Steiner forest problem \cite{HajiaghayiJ}. Thus, one is
forced to rely on alternative and problem-specific techniques.

\subsection{Overview of Technical Ideas}
We consider the rooted version of \kv: the goal is to find a min-cost
subgraph that $2$-connects at least $k$ terminals to a specified root
vertex $r$. It is relatively straightforward to reduce \kv to its
rooted version (see section~\ref{sec:k2vc} for details.)  We draw
inspiration from algorithmic ideas that led to poly-logarithmic
approximations for the $k$-MST problem. 

To describe our approach to the rooted \kv problem, we define a
closely related problem. For a subgraph $H$ that contains $r$, let
$k(H)$ be the number of terminals that are $2$-connected to $r$ in
$H$. Then the {\em density} of $H$ is simply the ratio of the cost of
$H$ to $k(H)$. The \densV problem is to find a 2-connected subgraph of
minimum density. An $O(\log \ell)$ approximation for the \densV
problem (where $\ell$ is the number of terminals) can be derived in a
some what standard way by using a bucketing and scaling trick on a
linear programming relaxation for the problem. We exploit the known
bound of $2$ on the integrality gap of a natural LP for the SNDP
problem with vertex connectivity requirements in $\{0,1,2\}$
\cite{FleischerJW}.  The bucketing and scaling trick has seen several
uses in the past and has recently been highlighted in several
applications \cite{ChekuriHKS06,ChekuriHKS07,ChekuriEGS08}.

Our algorithm for \kv uses a greedy approach at the high level. We
start with an empty subgraph $G'$ and use the approximation algorithm
for \densV in an iterative fashion to greedily add terminals to $G'$
until at least $k' \ge k$ terminals are in $G'$. This approach would
yield an $O(\log \ell \log k)$ approximation if $k' = O(k)$. However,
the last iteration of the \densV algorithm may add many more terminals
than desired with the result that $k' \gg k$. In this case we cannot
bound the quality of the solution obtained by the algorithm.  To
overcome this problem, one can try to {\em prune} the subgraph $H$
added in the last iteration to only have the desired number of
terminals. For the $k$-MST problem, $H$ is a tree and pruning is
quite easy. We remark that this yields a rather straightforward
$O(\log n \log k)$ approximation for $k$-MST and could have been
discovered much before a more clever analysis given in
\cite{AwerbuchABV95}.

One of our technical contributions is to give a pruning step for the
\kv problem. To accomplish this, we use two algorithmic ideas. The
first is encapsulated in the cycle finding algorithm of
Theorem~\ref{thm:cycle}.  Second, we use this cycle finding algorithm
to repeatedly merge subgraphs until we get the desired number of
terminals in one subgraph. This latter step requires care. The cycle
merging scheme is inspired by a similar approach from the work of Lau
\etal \cite{LauNSS07} on the \ke problem and in \cite{ChekuriKP08} on
the directed orienteering problem. These ideas yield an $O(\log \ell
\cdot \log^2 k)$ approximation. We give a slightly modified
cycle-merging algorithm with a more sophisticated and non-trivial
analysis to obtain an improved $O(\log \ell \cdot \log k)$
approximation.

Some remarks are in order to compare our work to that of
\cite{LauNSS07} on the \ke problem. The combinatorial algorithm in
\cite{LauNSS07} is based on finding a low-density
cycle or a related structure called a bi-cycle. The algorithm in
\cite{LauNSS07} to find such a structure is incorrect.  Further, the
cycles are contracted along the way which limits the approach to the
\ke problem (contracting a cycle in $2$-node-connected graph may make
the resulting graph not $2$-node-connected).  In our algorithm we do
not contract cycles and instead introduce dummy terminals with weights
to capture the number of terminals in an already formed component.
This requires us to now address the minimum-density non-trivial simple
cycle problem which we do via Theorem~\ref{thm:cycle} and
Corollary~\ref{cor:cycle}. In independent work, Lau \etal
\cite{LauNSS07b} obtain a new and correct $O(\log n \log
k)$-approximation for \ke.  They also follow the same approach that we
do in using the LP for finding dense subgraphs followed by the pruning
step.  However, in the pruning step they use a completely different
approach; they use the sophisticated idea of no-where zero $6$-flows
\cite{Seymour}. Although the use of this idea is elegant, the approach
works only for the \ke problem, while our approach is less complex
and leads to an algorithm for the more general \kv problem.

\section{The Algorithm for the \kv Problem}\label{sec:k2vc}

We work with graphs in which some vertices are designated as
\emph{terminals}.  Given a graph $G$ with edge costs and terminal
weights, we define the \emph{density} of a subgraph $H$ to be sum of
the costs of edges in $H$ divided by the sum of the weights of
terminals in $H$.  Henceforth, we use $2$-connected graph to mean a
$2$-vertex-connected graph.

The goal of the \kv problem is to find a minimum-cost 2-connected
subgraph on at least $k$ terminals.\footnote{In fact, our algorithm
  solves the harder problem in which terminals have weights, and the
  goal is to find a minimum-cost 2-connected subgraph in which the sum
  of terminal weights is at least $k$. For simplicity of exposition,
  however, we stick to the more restricted version.}  Recall that in
the rooted \kv problem, the goal is to find a min-cost subgraph on at
least $k$ terminals in which every terminal is 2-connected to the
specified root $r$. The (unrooted) \kv problem can be reduced to the
rooted version by \emph{guessing} 2 vertices $u,v$ that are in an
optimal solution, creating a new root vertex $r$, and connecting it
with 0-cost edges to $u$ and $v$. It is not hard to show that any
solution to the rooted problem in the modified graph can be converted
to a solution to the unrooted problem by adding 2 minimum-cost
vertex-disjoint paths between $u$ and $v$.  (Since $u$ and $v$ are in
the optimal solution, the cost of these added paths cannot be more
than $\opt$.) We omit further details from this extended abstract.

In the \densV problem, the goal is to find a subgraph $H$ of minimum
density in which all terminals of $H$ are 2-connected to the root. The
following lemma is proved in Section~\ref{subsec:LP} below.  It relies
on a $2$-approximation, via a natural LP, for the min-cost
$2$-connectivity problem due to Fleischer, Jain and Williamson
\cite{FleischerJW}, and some standard techniques.

\begin{lemma}\label{lem:densV}
  There is an $O(\log \ell)$-approximation algorithm for the \densV
  problem, where $\ell$ is the number of terminals in the given
  instance.
\end{lemma}

Let $\opt$ be the cost of an optimal solution to the \kv problem. We
assume knowledge of $\opt$; this can be dispensed with using standard
methods.  We pre-process the graph by deleting any terminal that does
not have 2 vertex-disjoint paths to the root $r$ of total cost at most
$\opt$. The high-level description of the algorithm for the 
rooted \kv problem is given below.

\begin{algo}
  $k' \assign k$, $\quad G'$ is the empty graph. \\
  While ($k' > 0$): \+ \\
  Use the approximation algorithm for \densV to find a subgraph $H$ in $G$.\\
  If ($k(H) \le k'$): \+ \\
  $G' \assign G' \cup H$, $\quad k' \assign k' - k(H)$ \\
  Mark all terminals in $H$ as non-terminals. \- \\
  Else: \+ \\
  {\em Prune} $H$ to obtain $H'$ that contains $k'$ terminals. \\
  $G' = G' \cup H'$, $\quad k' \assign 0$ \- \- \\
  Output $G'$
\end{algo}

At the beginning of any iteration of the while loop, the graph contains a
solution to the \densV problem of density at most $\frac{\opt}{k'}$.
Therefore, the graph $H$ returned always has density at most $O(\log \ell)
\frac{\opt}{k'}$. If $k(H) \le k'$, we add $H$ to $G'$ and decrement $k'$;
we refer to this as the \emph{augmentation} step. Otherwise, we have a
graph $H$ of good density, but with too many terminals. In this case, we
prune $H$ to find a graph with the required number of terminals; this is
the \emph{pruning step}. A simple set-cover type argument shows the
following lemma:

\begin{lemma}\label{lem:greedy}
  If, at every augmentation step, we add a graph of density at most
  $O(\log \ell) \frac{\opt}{k'}$ (where $k'$ is the number of
  additional terminals that must be selected), the total cost of all
  the augmentation steps is at most $O(\log \ell \cdot \log k) \opt$.
\end{lemma}

Therefore, we now only have to bound the cost of the graph $H'$ added
in the pruning step; we prove the following theorem in
Section~\ref{sec:pruning}.

\begin{theorem} \label{thm:avekv} Let $\langle G,k \rangle$ be an
  instance of the rooted \kv problem with root $r$, such that every
  vertex of $G$ has $2$ vertex-disjoint paths to $r$ of total cost at
  most $L$, and such that $\dens{G} \le \rho$. There is a
  polynomial-time algorithm to find a solution to this instance of
  cost at most $O(\log k)\rho k + 2L$.
\end{theorem}

\noindent We can now prove our main result for the \kv problem,
Theorem~\ref{thm:kv}.

\begin{proofof}{Theorem~\ref{thm:kv}}
  Let $\opt$ be the cost of an optimal solution to the (rooted) \kv
  problem.  By Lemma~\ref{lem:greedy}, the total cost of the
  augmentation steps of our greedy algorithm is $O(\log \ell \cdot
  \log k) \opt$. To bound the cost of the pruning step, let $k'$ be
  the number of additional terminals that must be covered just prior
  to this step. The algorithm for the \densV problem returns a graph
  $H$ with $k(H) > k'$ terminals, and density at most $O(\log \ell)
  \frac{\opt}{k'}$. As a result of our pre-processing step, every
  vertex has 2 vertex-disjoint paths to $r$ of total cost at most
  $\opt$. Now, we use Theorem~\ref{thm:avekv} to prune $H$ and find a
  graph $H'$ with $k'$ terminals and cost at most $O(\log k)
  density(H) k' + 2 \opt \le O(\log \ell \cdot \log k) \opt + 2
  \opt$. Therefore, the total cost of our solution is $O(\log \ell
  \cdot \log k) \opt$.
\end{proofof}

\bigskip
It remains only to prove Lemma~\ref{lem:densV}, that there is an $O(\log
\ell)$-approximation for the \densV problem, and
Theorem~\ref{thm:avekv}, bounding the cost of the pruning step. We
prove the former in Section~\ref{subsec:LP} below. Before the latter
is proved in Section~\ref{sec:pruning}, we develop some tools in
Section~\ref{sec:cycles}; chief among these tools is
Theorem~\ref{thm:cycle}. 

\subsection{An $O(\log \ell)$-approximation for the \densV problem}
\label{subsec:LP}

Recall that the \densV~problem was defined as follows: Given a graph
$G(V,E)$ with edge-costs, a set $T \subseteq V$ of terminals, and a
root $r \in V(G)$, find a subgraph $H$ of minimum density, in which
every terminal of $H$ is 2-connected to $r$. (Here, the density of $H$
is defined as the cost of $H$ divided by the number of terminals it
contains, not including $r$.)  We describe an algorithm for \densV
that gives an $O(\log \ell)$-approximation, and sketch its proof. We
use an LP based approach and a bucketing and scaling trick (see
\cite{ChekuriEGS08,ChekuriHKS06,ChekuriHKS07} for applications of this
idea), and a constant-factor bound on the integrality gap of an LP for
SNDP with vertex-connectivity requirements in $\{0,1,2\}$
\cite{FleischerJW}.

\bigskip
We define \densLP as the following LP relaxation of \densV. For each
terminal $t$, the variable $y_t$ indicates whether or not $v$ is
chosen in the solution. (By normalizing $\sum_t y_t$ to 1, and
minimizing the sum of edge costs, we minimize the density.)
$\script{C}_t$ is the set of all simple cycles containing $t$ and the
root $r$; for any $C \in \script{C}_t$, $f_C$ indicates how much
`flow' is sent from $v$ to $r$ through $C$. (Note that a pair of
vertex-disjoint paths is a cycle; the flow along a cycle is 1 if we
can 2-connect $t$ to $r$ using the edges of the cycle.) The variable
$x_e$ indicates whether the edge $e$ is used by the solution.

\newpage
\[\min \sum_{e \in E} c(e) x_e \]
\vspace{-0.15in}
\begin{align*}
  \sum_{t \in T}y_t & = 1 & \\
  \sum_{C \in \script{C}_t}f_C & \ge y_t & \left( \forall t \in T \right) \\
  \sum_{C \in \script{C}_t| e \in C} f_C & \le x_e & \left( \forall t \in T, e  \in E \right) \\
  x_e, f_c, y_t & \ge 0 & \\
\end{align*}

\vspace{-0.15in}
It is not hard to see that an optimal solution to \densLP has cost at
most the density of an optimal solution to \densV. We now show how to
obtain an integral solution of density at most $O(\log \ell)
\opt_{LP}$, where $\opt_{LP}$ is the cost of an optimal solution to
\densLP.  The linear program \densLP has an exponential number of
variables but a polynomial number of non-trivial constraints; it can,
however, be solved in polynomial time. Fix an optimal solution to
\densLP of cost $\opt_{LP}$, and for each $0 \le i < 2 \log \ell$ (for
ease of notation, assume $\log \ell$ is an integer), let $Y_i$ be the
set of terminals $t$ such that $2^{-(i+1)} < y_t \le 2^{-i}$. Since
$\sum_{t \in T}y_t = 1$, there is some index $i$ such that $\sum_{t
  \in Y_i}y_t \ge \frac{1}{2 \log \ell}$. Since every terminal $t \in
Y_i$ has $y_t \le 2^{-i}$, the number of terminals in $Y_i$ is at
least $\frac{2^{i-1}}{\log \ell}$.  We claim that there is a subgraph
$H$ of $G$ with cost at most $O(2^{i+2} \opt_{LP})$, in which every
terminal of $Y_i$ is 2-connected to the root. If this is true, the
density of $H$ is at most $O(\log \ell \cdot \opt_{LP})$, and hence we
have an $O(\log \ell)$-approximation for the \densV problem.

To prove our claim about the cost of the subgraph $H$ in which every
terminal of $Y_i$ is 2-connected to $r$, consider scaling up the given
optimum solution of \densLP by a factor of $2^{i+1}$. For each
terminal $t \in Y_i$, the flow from $t$ to $r$ in this scaled
solution\footnote{This is an abuse of the term `solution', since after
  scaling, $\sum_{t \in T} y_t = 2^{i+1}$} is at least 1, and the cost
of the scaled solution is $2^{i+1} \opt_{LP}$.

In \cite{FleischerJW}, the authors describe a linear program $LP_1$ to
find a minimum-cost subgraph in which a given set of terminals is
2-connected to the root, and show that this linear program has an
integrality gap of 2. The variables $x_e$ in the `scaled solution' to
\densLP correspond to a feasible solution of $LP_1$ with $Y_i$ as the
set of terminals; the integrality gap of 2 implies that there is a
subgraph $H$ in which every terminal of $Y_i$ is 2-connected to the
root, with cost at most $2^{i+2} \opt_{LP}$.

Therefore, the algorithm for \densV is: 
\begin{enumerate}
  \item Find an optimal fractional solution to \densLP.

  \item Find a set of terminals $Y_i$ such that $\sum_{t \in Y_i} y_t
    \ge \frac{1}{2 \log \ell}$.

  \item Find a min-cost subgraph $H$ in which every terminal in $Y_i$
    is 2-connected to $r$ using the algorithm of \cite{FleischerJW}.
    $H$ has density at most $O(\log \ell)$ times the optimal solution
    to \densV.
\end{enumerate}

\section{Finding Low-density Non-trivial Cycles}
\label{sec:cycles}

A cycle $C \subseteq G$ is \emph{non-trivial} if it contains at least 2
terminals.  We define the min-density non-trivial cycle problem: Given a
graph $G(V,E)$, with $S \subseteq V$ marked as terminals, edge costs and
terminal weights, find a minimum-density cycle that contains at least 2
terminals. Note that if we remove the requirement that the cycle be
non-trivial (that is, it contains at least 2 terminals), the problem
reduces to the min-mean cycle problem in directed graphs, and can be solved
exactly in polynomial time (see \cite{networkflows_book}).  Algorithms for
the min-density non-trivial cycle problem are a useful tool for solving the
\kv and \ke problems. In this section, we give an $O(\log
\ell)$-approximation algorithm for the minimum-density non-trivial cycle
problem.

First, we prove Theorem~\ref{thm:cycle}, that a 2-connected graph with
edge costs and terminal weights contains a simple non-trivial cycle,
with density no more than the average density of the graph. We give
two algorithms to find such a cycle; the first, described in
Section~\ref{subsec:nonPoly}, is simpler, but the running time is not
polynomial.  A more technical proof that leads to a strongly
polynomial-time algorithm is described in Section~\ref{subsec:strong};
we recommend this proof be skipped on a first reading.

\subsection{An Algorithm to Find Cycles of Average
  Density}\label{subsec:nonPoly}

To find a non-trivial cycle of density at most that of the 2-connected
input graph $G$, we will start with an arbitrary non-trivial cycle,
and successively find cycles of better density until we obtain a cycle
with density at most $\dens{G}$. The following lemma shows that if a
cycle $C$ has an ear with density less than $\dens{C}$, we can use
this ear to find a cycle of lower density.

\begin{lemma}\label{lem:goodEar}
  Let $C$ be a non-trivial cycle, and $H$ an ear incident to $C$ at
  $u$ and $v$, such that $\frac{cost(H)}{weight(H - \{u,v\})} <
  \dens{C}$. Let $S_1$ and $S_2$ be the two internally disjoint paths
  between $u$ and $v$ in $C$. Then $H \cup S_1$ and $H \cup S_2$ are
  both simple cycles and one of these is non-trivial and has density
  less than $\dens{C}$.
\end{lemma}
\begin{proof}
  $C$ has at least 2 terminals, so it has finite density; $H$ must
  then have at least 1 terminal. Let $c_1$, $c_2$ and $c_H$ be,
  respectively, the sum of the costs of the edges in $S_1$, $S_2$ and
  $H$, and let $w_1$, $w_2$ and $w_H$ be the sum of the weights of the
  terminals in $S_1$, $S_2$ and $H-\{u,v\}$.

  Assume w.l.o.g. that $S_1$ has density at most that of $S_2$. (That
  is, $c_1/w_1 \le c_2/w_2$.)\footnote{It is possible that one of
    $S_1$ and $S_2$ has cost 0 and weight 0. In this case, let $S_1$
    be the component with non-zero weight.}  $S_1$ must contain at
  least one terminal, and so $H \cup S_1$ is a simple non-trivial
  cycle.  The statement $\dens{H \cup S_1} < \dens{C}$ is equivalent
  to $(c_H + c_1)(w_1 + w_2) < (c_1 + c_2) (w_H + w_1)$.
  \vspace{-0.2in}

  \begin{align*}
    (c_H + c_1)(w_1 + w_2) &= c_1w_1 + c_1w_2 + c_H(w_1 + w_2)\\
    &\le c_1w_1 + c_2w_1 + c_H(w_1 + w_2) & (\dens{S_1} \le \dens{S_2})\\
    &< c_1w_1 + c_2w_1 + (c_1 + c_2)w_H   & (c_H/w_H < \dens{C})\\
    &= (c_1 + c_2) (w_H + w_1)\\
  \end{align*}

  \vspace{-0.2in}
  Therefore, $H \cup S_1$ is a simple cycle containing at least 2 terminals of
  density less than $\dens{C}$.
\end{proof}

\begin{lemma}\label{lem:2connComp}
  Given a cycle $C$ in a $2$-connected graph $G$, let $G'$ be the graph
  formed from $G$ by contracting $C$ to a single vertex $v$. If $H$ is
  a connected component of $G' - v$, $H \cup \{v\}$ is $2$-connected in
  $G'$.
\end{lemma}
\begin{proof}
  Let $H$ be an arbitrary connected component of $G' - v$, and let $H'
  = H \cup \{v\}$. To prove that $H'$ is 2-connected, we first observe
  that $v$ is 2-connected to any vertex $x \in H$. (Any set that
  separates $x$ from $v$ in $H'$ separates $x$ from the cycle $C$ in
  $G$.)

  It now follows that for all vertices $x,y \in V(H)$, $x$ and $y$ are
  2-connected in $H'$. Suppose deleting some vertex $u$ separates $x$
  from $y$. The vertex $u$ cannot be $v$, since $H$ is a connected
  component of $G' - v$. But if $u \neq v$, $v$ and $x$ are in the
  same component of $H' - u$, since $v$ is 2-connected to $x$ in
  $H'$. Similarly, $v$ and $y$ are in the same component of $H' - u$,
  and so deleting $u$ does not separate $x$ from $y$.
\end{proof}

We now show that given any 2-connected graph $G$, we can find a
non-trivial cycle of density no more than that of $G$.

\begin{theorem}\label{thm:cycleExists} 
  Let $G$ be a $2$-connected graph with at least $2$ terminals. $G$
  contains a simple non-trivial cycle $X$ such that $\dens{X} \le
  \dens{G}$.
\end{theorem}
\begin{proof}
  Let $C$ be an arbitrary non-trivial simple cycle; such a cycle
  always exists since $G$ is $2$-connected and has at least 2
  terminals.  If $\dens{C} > \dens{G}$, we give an algorithm that
  finds a new non-trivial cycle $C'$ such that $\dens{C'} < \dens{C}$.
  Repeating this process, we obtain cycles of successively better
  densities until eventually finding a non-trivial cycle $X$ of
  density at most $\dens{G}$.

  Let $G'$ be the graph formed by contracting the given cycle $C$ to a
  single vertex $v$. In $G'$, $v$ is not a terminal, and so has weight
  0. Consider the 2-connected components of $G'$ (from
  Lemma~\ref{lem:2connComp}, each such component is formed by adding
  $v$ to a connected component of $G' - v$), and pick the one of
  minimum density. If $H$ is this component, $\dens{H} < \dens{G}$ by
  an averaging argument.

  $H$ contains at least 1 terminal. If it contains 2 or more
  terminals, recursively find a non-trivial cycle $C'$ in $H$ such
  that $\dens{C'} \le \dens{H} < \dens{C}$. If $C'$ exists in the
  given graph $G$, it has the desired properties, and we are
  done. Otherwise, $C'$ contains $v$, and the edges of $C'$ form a ear
  of $C$ in the original graph $G$. The density of this ear is less
  than the density of $C$, so we can apply Lemma~\ref{lem:goodEar} to
  obtain a non-trivial cycle in $G$ that has density less than
  $\dens{C}$.

  Finally, if $H$ has exactly 1 terminal $u$, find any 2
  vertex-disjoint paths using edges of $H$ from $u$ to distinct
  vertices in the cycle $C$. (Since $G$ is 2-connected, there always
  exist such paths.) The cost of these paths is at most $cost(H)$, and
  concatenating these 2 paths corresponds to a ear of $C$ in $G$.  The
  density of this ear is less than $\dens{C}$; again, we use
  Lemma~\ref{lem:goodEar} to obtain a cycle in $G$ with the desired
  properties.
\end{proof}

We remark again that the algorithm of Theorem~\ref{thm:cycleExists} does
not lead to a polynomial-time algorithm, even if all edge costs and
terminal weights are polynomially bounded. In Section~\ref{subsec:strong},
we describe a strongly polynomial-time algorithm that, given a graph $G$,
finds a non-trivial cycle of density at most that of $G$.
Note that neither of these algorithms may directly give a good
approximation to the min-density non-trivial cycle problem, because
the optimal non-trivial cycle may have density much less than that of
$G$.  However, we can use Theorem~\ref{thm:cycleExists} to prove the
following theorem:

\begin{theorem}\label{thm:equivalence}
  There is an $\alpha$-approximation to the (unrooted) \densV problem
  if and only if there is an $\alpha$-approximation to the problem of
  finding a minimum-density non-trivial cycle.
\end{theorem}
\begin{proof}
  Assume we have a $\gamma(\ell)$-approximation for the \densV
  problem; we use it to find a low-density non-trivial cycle. Solve
  the \densV problem on the given graph; since the optimal cycle is a
  2-connected graph, our solution $H$ to the \densV problem has
  density at most $\gamma(\ell)$ times the density of this cycle.
  Find a non-trivial cycle in $H$ of density at most that of $H$; it
  has density at most $\gamma(\ell)$ times that of an optimal
  non-trivial cycle.

  Note that any instance of the (unrooted) \densV problem has an
  optimal solution that is a non-trivial cycle.  (Consider any optimal
  solution $H$ of density $\rho$; by Theorem~\ref{thm:cycle}, $H$
  contains a non-trivial cycle of density at most $\rho$. This cycle
  is a valid solution to the \densV problem.) Therefore, a
  $\beta(\ell)$-approximation for the min-density non-trivial cycle
  problem gives a $\beta(\ell)$-approximation for the \densV problem.
\end{proof}

Theorem~\ref{thm:equivalence} and Lemma~\ref{lem:densV} imply 
an $O(\log \ell)$-approximation for the minimum-density non-trivial cycle
problem; this proves Corollary~\ref{cor:cycle}.

We say that a graph $G(V,E)$ is minimally 2-connected on its terminals
if for every edge $e \in E$, some pair of terminals is not 2-connected
in the graph $G - e$.  Section~\ref{subsec:strong} shows that in any
graph which is minimally 2-connected on its terminals, every cycle is
non-trivial.  Therefore, the problem of finding a minimum-density
non-trivial cycle in such graphs is just that of finding a
minimum-density cycle, which can be solved exactly in polynomial
time. However, as we explain at the end of the section, this does not
directly lead to an efficient algorithm for arbitrary graphs.

\subsection{A Strongly Polynomial-time Algorithm to Find Cycles of Average
  Density}\label{subsec:strong}

In this section, we describe a strongly polynomial-time algorithm
which, given a 2-connected graph $G(V,E)$ with edge costs and terminal
weights, finds a non-trivial cycle of density at most that of $G$.

We begin with several definitions: Let $C$ be a cycle in a graph $G$,
and $G'$ be the graph formed by deleting $C$ from $G$. Let $H_1, H_2,
\ldots H_m$ be the connected components of $G'$; we refer to these as
\emph{earrings} of $C$.\footnote{If $H_i$ were simply a path, it would be an
  ear of $C$, but $H_i$ may be more complex.}  For each $H_i$, let the
vertices of $C$ incident to it be called its \emph{clasps}. From the
definition of an earring, for any pair of clasps of $H_i$, there is a
path between them whose internal vertices are all in $H_i$. 

We say that a vertex of $C$ is an \emph{anchor} if it is the clasp of
some earring. (An anchor may be a clasp of multiple earrings.)  A
\emph{segment} $S$ of $C$ is a path contained in $C$, such that the
endpoints of $S$ are both anchors, and no internal vertex of $S$ is an
anchor. (Note that the endpoints of $S$ might be clasps of the same
earring, or of distinct earrings.) It is easy to see that the segments
partition the edge set of $C$. By deleting a segment, we refer to
deleting its edges and internal vertices. Observe that if $S$ is
deleted from $G$, the only vertices of $G - S$ that lose an edge are
the endpoints of $S$. A segment is \emph{safe} if the graph $G-S$ is
2-connected.

Arbitrarily pick a vertex $o$ of $C$ as the \emph{origin}, and
consecutively number the vertices of $C$ clockwise around the cycle as
$o = c_0, c_1, c_2, \ldots, c_r = o$. The first clasp of an earring
$H$ is its lowest numbered clasp, and the last clasp is its highest
numbered clasp. (If the origin is a clasp of $H$, it is considered the
first clasp, not the last.) The \emph{arc} of an earring is the
subgraph of $C$ found by traversing clockwise from its first clasp
$c_p$ to its last clasp $c_q$; the length of this arc is $q-p$. (That
is, the length of an arc is the number of edges it contains.) Note
that if an arc contains the origin, it must be the first vertex of the
arc.  Figure~\ref{fig:earring} illustrates several of these
definitions.

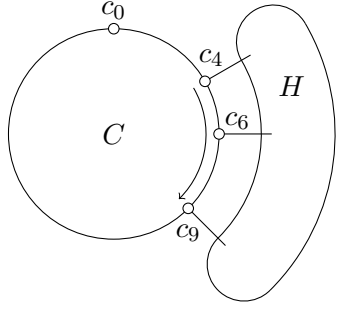
\begin{figure}
  \begin{center}
    \begin{tikzpicture}[scale=0.7]
      
      
      \draw (0,0) circle (2cm); \node at (0,0) {$C$};
      \draw (30:2cm) -- (30:3cm); \draw (2,0) -- (3,0); \draw
      (-45:2cm) -- (-45:3cm);
      \node at (38:2.35cm) {$c_4$}; \node at (2.35,0.25) {$c_6$}; \node at (-53:2.35cm) {$c_9$};
      \node at (0,2.35) {$c_0$};

       \draw (30:2.8cm) arc(30:-45:2.8cm); \draw (30:4.2cm) arc (30:-45:4.2cm);
       \draw (30:2.8cm) arc (-150:-330:0.7cm);
       \draw (-45:2.8cm) arc (135:315:0.7cm);
       \node at (15:3.5cm) {$H$};

       \draw[fill=white] (30:2cm) circle (1mm); \draw[fill=white] (-45:2cm) circle (1mm);
       \draw[fill=white] (2,0) circle (1mm); \draw[fill=white] (0,2) circle (1mm);

       \draw[->] (30:1.75cm) arc (30:-45:1.75cm);
    \end{tikzpicture}
  \end{center}
\caption{$H$ is an earring of $G$, with clasps $c_4, c_6, c_9$; $c_4$
  is its first clasp, and $c_9$ its last clasp. The arrow indicates
  the arc of $H$.}\label{fig:earring}
\end{figure}

\begin{figure}[tbh]
  \begin{center}
    \begin{tikzpicture}[scale=0.75]
      
      
      \begin{scope}[rotate=-60]
        \draw (0,0) circle (2cm); \node at (0,0) {$C$};
        \draw[red] (0,2) -- (0,3.6) arc(90:-90:3.6cm) -- (0,-2);
        \node at (0,1.6) {$c_a$}; \node at (0,-1.6) {$c_b$};
        \node at (0,3.9) [font=\footnotesize] {$H \in \script{H}$};
        
        \draw[blue] (30:2cm) -- (30:3cm) arc(30:-30:3cm) -- (-30:2cm);
        \node at (30:1.6cm) {$c_p$}; \node at (-30:1.6cm) {$c_q$};
        
        \draw[dotted] (30:2.2cm) arc(30:330:2.2cm);
        \draw[dashed] (25:2.2cm) -- (25:2.8cm) arc(25:-25:2.8cm) -- (-25:2.2cm);
        
        \draw[fill=white] (30:2cm) circle (1mm); \draw[fill=white] (-30:2cm) circle (1mm);
        \draw[fill=white] (0,2) circle (1mm); \draw[fill=white] (0,-2) circle (1mm);
        \draw[fill=white] (150:2cm) circle (1mm); \node at (150:1.6cm) {$c_0$};
      \end{scope}

      \begin{scope}[xshift=8cm,rotate=-60]

        \draw (0,0) circle (2cm); \node at (0,0) {$C$};
        \draw[red] (0,2) -- (0,3.6) arc(90:-90:3.6cm) -- (0,-2);
        \node at (0,1.6) {$c_a$}; \node at (0,-1.6) {$c_b$};
        \node at (0,3.9) [font=\footnotesize] {$H_1$};
        \node at (30:1.6cm) {$c_p$}; \node at (-30:1.6cm) {$c_q$};

        \draw[red] (30:2cm) -- (30:3.6cm);

        \draw[blue] (88:2cm) -- (85:3cm) arc(85:-85:3cm) -- (-88:2cm);
        \node at (85:3.3cm) [font=\footnotesize] {$H_2$};
        \draw[blue] (-30:2cm) -- (-30:3cm);

        \draw[dotted] (25:2.2cm) -- (25:3.4cm) arc (25:-88:3.4cm) --
        (-88:2.2cm) arc (-88:-35:2.2cm);
        \draw[dashed] (-25:2.2cm) -- (-25:2.8cm) arc (-25:80:2.8cm) --
        (80:2.2cm) arc (80:35:2.2cm);

        \draw[fill=white] (30:2cm) circle (1mm); \draw[fill=white] (-30:2cm) circle (1mm);
        \draw[fill=white] (0,2) circle (1mm); \draw[fill=white] (0,-2) circle (1mm);
        \draw[fill=white] (150:2cm) circle (1mm); \node at (150:1.6cm) {$c_0$};
      \end{scope}

      \begin{scope}[xshift=16cm,rotate=-60]

        \draw (0,0) circle (2cm); \node at (0,0) {$C$};
        \draw[red] (0,2) -- (0,3.6) arc(90:-90:3.6cm) -- (0,-2);
        \node at (0,1.6) {$c_a$}; \node at (0,-1.6) {$c_b$};
        \node at (0,3.9) [font=\footnotesize] {$H \in \script{H}$};
        \node at (30:1.6cm) {$c_p$}; \node at (-30:1.6cm) {$c_q$};

        \draw[blue] (30:2cm) -- (30:3.1cm) arc (30:-105:3.1cm) -- (-105:2cm);
        \node at (38:2.55cm) [font=\footnotesize] {$H_1$};
        \draw[green] (-30:2cm) -- (-30:2.5cm) arc(-30:-135:2.5cm) -- (-135:2cm);
        \node at (-143:2.5cm) [font=\footnotesize] {$H_2$};

        \draw[dotted] (35:2.15cm) arc (35:85:2.15cm) -- (85:3.45cm) arc
        (85:-85:3.45cm) -- (-85:2.15cm) arc (-85:-35:2.15cm);
        \draw[dashed] (25:2.2cm) -- (25:2.95cm) arc(25:-100:2.95cm) --
        (-100:2.15cm) arc (-100:-130:2.15cm) -- (-130:2.65cm) arc
        (-130:-25:2.65cm) -- (-25:2.15cm);

        \draw[fill=white] (30:2cm) circle (1mm); \draw[fill=white] (-30:2cm) circle (1mm);
        \draw[fill=white] (0,2) circle (1mm); \draw[fill=white] (0,-2) circle (1mm);
        \draw[fill=white] (150:2cm) circle (1mm); \node at (150:1.6cm) {$c_0$};
      \end{scope}

      \begin{scope}[xshift=4cm,yshift=-7cm,rotate=-60]

        \draw (0,0) circle (2cm); \node at (0,0) {$C$};
        \draw[red] (0,2) -- (0,3.6) arc(90:-90:3.6cm) -- (0,-2);
        \node at (0,1.6) {$c_a$}; \node at (0,-1.6) {$c_b$};
        \node at (0,4.1) [font=\footnotesize] {$H_1 \in \script{H}$};
        \node at (30:1.6cm) {$c_p$}; \node at (-30:1.6cm) {$c_q$};
        \draw[red] (30:2cm) -- (30:3.6cm);

        \draw[blue] (-30:2cm) -- (-30:3cm) arc(-30:-120:3cm) -- (-120:2cm);
        \node at (-125:3cm) [font=\footnotesize] {$H_2$};
        
        \draw[dashed] (-35:2.2cm) arc (-35:-85:2.2cm) -- (-85:3.4cm)
        arc (-85:25:3.4cm) -- (25:2.2cm);
        \draw[dotted] (-35:2.2cm) -- (-35:2.8cm) arc(-35:-115:2.8cm)
        -- (-115:2.2cm) arc (-115:-330:2.2cm);

        \draw[fill=white] (30:2cm) circle (1mm); \draw[fill=white] (-30:2cm) circle (1mm);
        \draw[fill=white] (0,2) circle (1mm); \draw[fill=white] (0,-2) circle (1mm);
        \draw[fill=white] (150:2cm) circle (1mm); \node at (150:1.6cm) {$c_0$};
      \end{scope}

      \begin{scope}[xshift=12cm,yshift=-7cm,rotate=-60]

        \draw (0,0) circle (2cm); \node at (0,0) {$C$};
        \draw[red] (0,2) -- (0,3.6) arc(90:-90:3.6cm) -- (0,-2);
        \node at (0,1.6) {$c_a$}; \node at (0,-1.6) {$c_b$};
        \node at (0,4.1) [font=\footnotesize] {$H_1 \in \script{H}$};
        \node at (30:1.6cm) {$c_p$}; \node at (-30:1.6cm) {$c_q$};
        \draw[red] (-30:2cm) -- (-30:3.6cm);

        \draw[blue] (30:2cm) -- (30:3cm) arc(30:-120:3cm) -- (-120:2cm);
        \node at (-125:3cm) [font=\footnotesize] {$H_2$};
        
        \draw[dashed] (-25:2.2cm) -- (-25:3.4cm) arc (-25:85:3.4cm) -- (85:2.2cm)
        arc (85:35:2.2cm);
        \draw[dotted] (-35:2.2cm) arc (-35:-115:2.2cm) -- (-115:2.8cm) arc(-115:35:2.8cm)
        -- (35:2.2cm);

        \draw[fill=white] (30:2cm) circle (1mm); \draw[fill=white] (-30:2cm) circle (1mm);
        \draw[fill=white] (0,2) circle (1mm); \draw[fill=white] (0,-2) circle (1mm);
        \draw[fill=white] (150:2cm) circle (1mm); \node at (150:1.6cm) {$c_0$};
      \end{scope}

    \end{tikzpicture}
  \end{center}
  \caption{The various cases of Theorem~\ref{thm:earringProof} are
    illustrated in the order presented. In each case, one of the 2
    vertex-disjoint paths from $c_p$ to $c_q$ is indicated with dashed
    lines, and the other with dotted lines.}\label{fig:earringProof}
\end{figure}
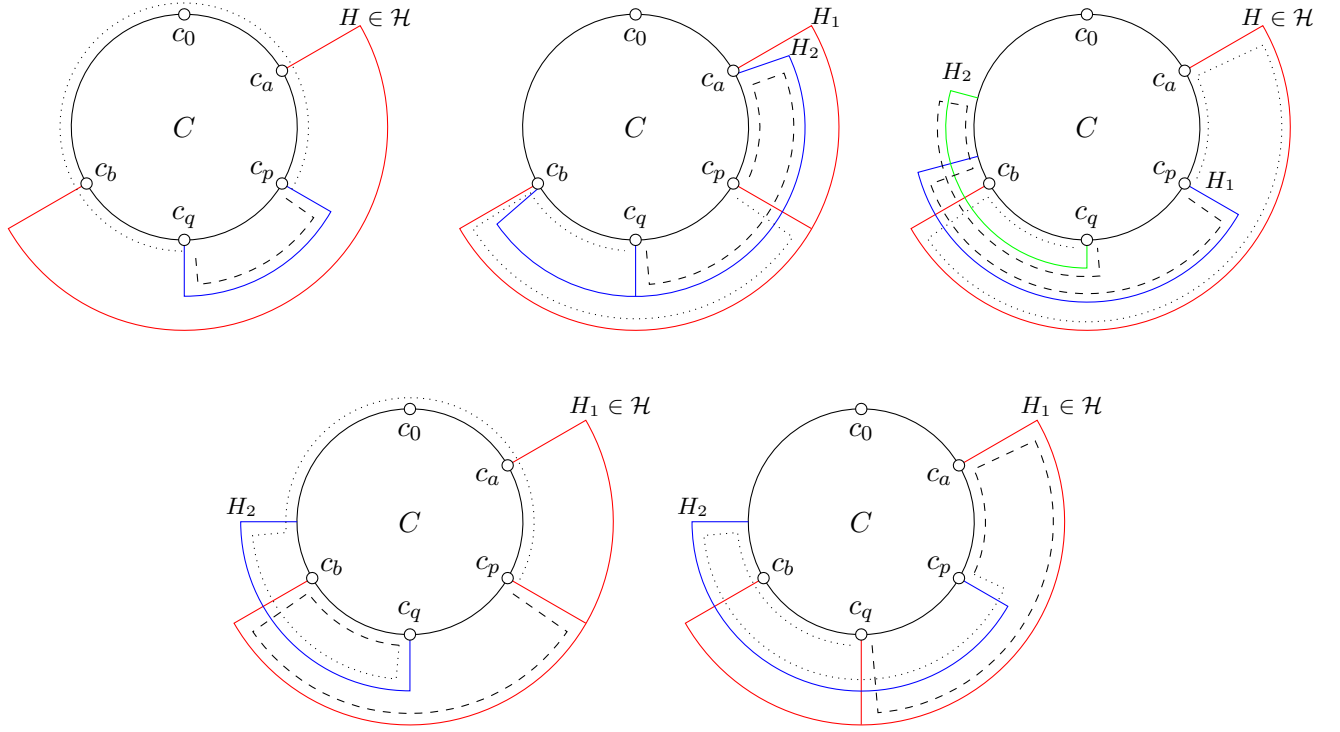

\begin{theorem}\label{thm:earringProof}
Let $H$ be an earring of minimum arc length. Every segment contained
in the arc of $H$ is safe.
\end{theorem}
\begin{proof}
  Let $\script{H}$ be the set of earrings with arc identical to that
  of $H$. Since they have the same arc, we refer to this as the arc of
  $\script{H}$, or the \emph{critical arc}. Let the first clasp of
  every earring in $\script{H}$ be $c_a$, and the last clasp of each
  earring in $\script{H}$ be $c_b$. Because the earrings in
  $\script{H}$ have arcs of minimum length, any earring $H' \notin
  \script{H}$ has a clasp $c_x$ that is not in the critical arc. (That
  is, $c_x < c_a$ or $c_x > c_b$.)

  We must show that every segment contained in the critical arc is
  safe; recall that a segment $S$ is safe if the graph $G-S$ is
  2-connected.  Given an arbitrary segment $S$ in the critical arc,
  let $c_p$ and $c_q$ ($p < q$) be the anchors that are its
  endpoints. We prove that there are always 2 internally
  vertex-disjoint paths between $c_p$ and $c_q$ in $G - S$; this
  suffices to show 2-connectivity.

  We consider several cases, depending on the earrings that contain
  $c_p$ and $c_q$. Figure~\ref{fig:earringProof} illustrates these
  cases.  If $c_p$ and $c_q$ are contained in the same earring $H'$,
  it is easy to find two vertex-disjoint paths between them in
  $G-S$. The first path is clockwise from $q$ to $p$ in the cycle
  $C$. The second path is entirely contained in the earring $H'$ (an
  earring is connected in $G-C$, so we can always find such a path.)

  Otherwise, $c_p$ and $c_q$ are clasps of distinct earrings. We
  consider three cases: Both $c_p$ and $c_q$ are clasps of earrings in
  $\script{H}$, one is (but not both), or neither is.
  \begin{enumerate}
  \item We first consider that both $c_p$ and $c_q$ are clasps of
    earrings in $\script{H}$. Let $c_p$ be a clasp of $H_1$, and $c_q$
    a clasp of $H_2$. The first path is from $c_q$ to $c_a$ through
    $H_2$, and then clockwise along the critical arc from $c_a$ to
    $c_p$. The second path is from $c_q$ to $c_b$ clockwise along the
    critical path, and then $c_b$ to $c_p$ through $H_1$. It is easy
    to see that these paths are internally vertex-disjoint.

  \item Now, suppose neither $c_p$ nor $c_q$ is a clasp of an earring
    in $\script{H}$. Let $c_p$ be a clasp of $H_1$, and $c_q$ be a
    clasp of $H_2$. The first path we find follows the critical arc
    clockwise from $c_q$ to $c_b$ (the last clasp of the critical
    arc), from $c_b$ to $c_a$ through $H \in \script{H}$, and again
    clockwise through the critical arc from $c_a$ to $c_p$. Internal
    vertices of this path are all in $H$ or on the critical arc. Let
    $c_{p'}$ be a clasp of $H_1$ not on the critical arc, and $c_{q'}$
    be a last clasp of $H_2$ not on the critical arc. The second path
    goes from $c_p$ to $c_{p'}$ through $H_1$, from $p'$ to $q'$
    through the cycle $C$ outside the critical arc, and from $c_{q'}$
    to $c_q$ through $H_2$. Internal vertices of this path are in
    $H_1, H_2$, or in $C$, but not part of the critical arc (since
    each of $c_{p'}$ and $c_{q'}$ are outside the critical arc).
    Therefore, we have 2 vertex-disjoint paths from $c_p$ to $c_q$.

  \item Finally, we consider the case that exactly one of $c_p, c_q$
    is a clasp of an earring in $\script{H}$. Suppose $c_p$ is a clasp
    of $H_1 \in \script{H}$, and $c_q$ is a clasp of $H_2 \notin
    \script{H}$; the other case (where $H_1 \notin \script{H}$ and
    $H_2 \in \script{H}$ is symmetric, and omitted, though
    figure~\ref{fig:earringProof} illustrates the paths.) Let $q'$ be
    the index of a clasp of $H_2$ outside the critical arc. The first
    path is from $c_q$ to $c_b$ through the critical arc, and then
    from $c_b$ to $c_p$ through $H_1$. The second path is from $c_q$
    to $c_{q'}$ through $H_2$, and from $c_{q'}$ to $c_p$ clockwise
    through $C$. Note that the last part of this path enters the
    critical arc at $c_a$, and continues through the arc until $c_p$.
    Internal vertices of the first path that are in $C$ are on the
    critical arc, but have index greater than $q$. Internal vertices
    of the second path that belong to $C$ are either not in the
    critical arc, or have index between $c_a$ and $c_p$. Therefore,
    the two paths are internally vertex-disjoint. \qedhere
  \end{enumerate}
\end{proof}

We now describe our algorithm to find a non-trivial cycle of good
density, proving Theorem~\ref{thm:cycle}:
  \emph{Let $G$ be a $2$-connected graph with edge-costs and terminal weights,
  and at least $2$ terminals. There is a polynomial-time algorithm to
  find a non-trivial cycle $X$ in $G$ such that $\dens{X} \le
  \dens{G}$.}

\begin{proofof}{Theorem~\ref{thm:cycle}}
  Let $G$ be a graph with $\ell$ terminals and density $\rho$; we describe
  a polynomial-time algorithm that either finds a cycle in $G$ of density
  less than $\rho$, or a proper subgraph $G'$ of $G$ that contains all $\ell$
  terminals. In the latter case, we can recurse on $G'$ until we eventually
  find a cycle of density at most $\rho$.

  We first find, in $O(n^3)$ time, a minimum-density cycle $C$ in $G$. By
  Theorem~\ref{thm:cycleExists}, $C$ has density at most $\rho$, because
  the minimum-density \emph{non-trivial} cycle has at most this density.
  If $C$ contains at least 2 terminals, we are done. Otherwise, $C$
  contains exactly one terminal $v$. Since $G$ contains at least 2
  terminals, there must exist at least one earring of $C$.

  Let $v$ be the origin of this cycle $C$, and $H$ an earring of minimum
  arc length. By Theorem~\ref{thm:earringProof}, every segment in the arc
  of $H$ is safe. Let $S$ be such a segment; since $v$ was selected as the
  origin, $v$ is not an internal vertex of $S$. As $v$ is the only terminal
  of $C$, $S$ contains no terminals, and therefore, the graph $G' = G - S$
  is 2-connected, and contains all $\ell$ terminals of $G$.
\end{proofof}

The proof above also shows that if $G$ is minimally 2-connected on its
terminals (that is, $G$ has no 2-connected proper subgraph containing
all its terminals), every cycle of $G$ is non-trivial. (If a cycle
contains 0 or 1 terminals, it has a safe segment containing no
terminals, which can be deleted; this gives a contradiction.)
Therefore, given a graph that \emph{is} minimally 2-connected on its
terminals, finding a minimum-density non-trivial cycle is equivalent
to finding a minimum-density cycle, and so can be solved exactly in
polynomial time.  This suggests a natural algorithm for the problem:
Given a graph that is not minimally 2-connected on its terminals,
delete edges and vertices until the graph is minimally 2-connected on
the terminals, and then find a minimum-density cycle. As shown above,
this gives a cycle of density no more than that of the input graph,
but this may not be the minimum-density cycle of the original
graph. For instance, there exist instances where the minimum-density
cycle uses edges of a safe segment $S$ that might be deleted by this
algorithm.

\section{Pruning 2-connected Graphs of Good Density}
\label{sec:pruning}

In this section, we prove Theorem~\ref{thm:avekv}. We are given a
graph $G$ and $S \subseteq V$, a set of at least $k$ terminals.
Further, every terminal in $G$ has 2 vertex-disjoint paths to the root
$r$ of total cost at most $L$. Let $\ell$ be the number of terminals
in $G$, and $cost(G)$ its total cost; $\rho = \frac{cost(G)}{\ell}$ is
the density of $G$. We describe an algorithm that finds a subgraph $H$
of $G$ that contains at least $k$ terminals, each of which is
2-connected to the root, and of total edge cost $O(\log k) \rho k + 2L$.

We can assume $\ell > (8 \log k)\cdot k$, or the trivial solution of
taking the entire graph $G$ suffices. The main phase of our algorithm
proceeds by maintaining a set of 2-connected subgraphs that we call
\emph{clusters}, and repeatedly finding low-density cycles that merge
clusters of similar weight to form larger clusters.  (The weight of a
cluster $X$, denoted by $w_X$, is (roughly) the number of terminals it
contains.) Clusters are grouped into \emph{tiers} by weight; tier $i$
contains clusters with weight at least $2^i$ and less than
$2^{i+1}$. Initially, each terminal is a separate cluster in tier
0. We say a cluster is \emph{large} if it has weight at least $k$, and
\emph{small} otherwise. The algorithm stops when most terminals are in
large clusters.

We now describe the algorithm {\sc MergeClusters} (see next page). To simplify
notation, let $\alpha$ be the quantity $2 \ceil{\log k} \rho$.  We say
that a cycle is \emph{good} if it has density at most $\alpha$; that
is, good cycles have density at most $O(\log k)$ times the density of
the input graph.

\begin{algo}
\underline{\sc MergeClusters}:\\
For (each $i$ in $\{0, 1, \dots, \left(\ceil{\log_2 k} - 1 \right)\}$) do: \+ \\
    If ($i = 0$): \+ \\
        Every terminal has weight 1 \- \\
    Else: \+ \\
        Mark all vertices as non-terminals \\
        For (each small 2-connected cluster $X$ in tier $i$) do: \+ \\
             Add a (dummy) terminal $v_X$ to $G$ of weight $w_X$ \\
             Add (dummy) edges of cost 0 from $v_X$ to two (arbitrary) distinct vertices of $X$ \- \- \\

    While ($G$ has a non-trivial cycle $C$ of density at most $\alpha$ in $G$): \+ \\
        Let $X_1, X_2, \ldots X_q$ be the small clusters that contain a terminal
           {\bf or an edge} of $C$.\\
        (Note that the terminals in $C$ belong to a subset of $\{X_1, \ldots X_q\}$.)\\
        Form a new cluster $Y$ (of a higher tier) by merging the clusters $X_1, \ldots X_q$\\
        $w_Y \assign \sum_{j=1}^q w_{X_j}$ \\
        If ($i = 0$): \+ \\
            Mark all terminals in $Y$ as non-terminals \- \\
        Else: \+ \\
            Delete all (dummy) terminals in $Y$ and the associated (dummy) edges.
\end{algo}

We briefly remark on some salient features of this algorithm and our
analysis before presenting the details of the proofs.
\begin{enumerate}
\item In iteration $i$, the terminals correspond to tier $i$
  clusters. Clusters are 2-connected subgraphs of $G$, and by using
  cycles to merge clusters, we preserve 2-connectivity as the clusters
  become larger.

\item When a cycle $C$ is used to merge clusters, all small clusters
  that contain an edge of $C$ (regardless of their tier) are merged to
  form the new cluster. Therefore, at any stage of the algorithm, all
  currently small clusters are edge-disjoint.  Large clusters, on the
  other hand, are \emph{frozen}; even if they intersect a good cycle
  $C$, they are not merged with other clusters on $C$. Thus, at any
  time, an edge may be in multiple large clusters and up to one small
  cluster.

\item In iteration $i$ of {\sc MergeClusters}, the density of a cycle
  $C$ is only determined by its cost and the weight of terminals in
  $C$ corresponding to tier $i$ clusters. Though small clusters of
  other (lower or higher) tiers might be merged using $C$, we do
  \emph{not} use their weight to pay for the edges of $C$.

\item The $i$th iteration terminates when no good cycles can be found
  using the remaining tier $i$ clusters. At this point, there may be
  some terminals remaining that correspond to clusters which are not
  merged to form clusters of higher tiers. However, our choice of
  $\alpha$ (which defines the density of good cycles) is such that we
  can bound the number of terminals that are ``left behind'' in this
  fashion. Therefore, when the algorithm terminates, most terminals
  are in large clusters.
\end{enumerate}

By bounding the density of large clusters, we can find a solution to
the rooted \kv problem of bounded density. Because we always use
cycles of low density to merge clusters, an analysis similar to that
of \cite{LauNSS07} and \cite{ChekuriKP08} shows that every large
cluster has density at most $O(\log^2 k) \rho$. We first present this
analysis, though it does not suffice to prove Theorem~\ref{thm:avekv}.
A more careful analysis shows that there is at least one large cluster
of density at most $O(\log k) \rho$; this allows us to prove the
desired theorem.

We now formally prove that {\sc MergeClusters} has the desired
behavior. First, we present a series of claims which, together, show
that when the algorithm terminates, most terminals are in large
clusters, and all clusters are 2-connected.

\begin{remark}\label{rem:cluster}
  Throughout the algorithm, the graph $G$ is always 2-connected. The
  weight of a cluster is at most the number of terminals it contains.
\end{remark}
\begin{proof}
  The only structural changes to $G$ are when new vertices are added
  as terminals; they are added with edges to two distinct vertices of
  $G$. This preserves 2-connectivity, as does deleting these terminals
  with the associated edges.

  To see that the second claim is true, observe that if a terminal
  contributes weight to a cluster, it is contained in that cluster. A
  terminal can be in multiple clusters, but it contributes to the
  weight of exactly one cluster.
\end{proof}

We use the following simple proposition in proofs of 2-connectivity;
the proof is straightforward, and hence omitted.

\begin{proposition}\label{prop:shareEdge}
  Let $H_1=(V_1,E_1)$ and $H_2=(V_2, E_2)$ be $2$-connected subgraphs
  of a graph $G(V,E)$ such that $|V_1 \cap V_2| \ge 2$.  Then the
  graph $H_1 \cup H_2 = (V_1 \cup V_2, E_1 \cup E_2)$ is
  $2$-connected.
\end{proposition}

\begin{lemma}\label{lem:clusters2conn}
The clusters formed by {\sc MergeClusters} are all $2$-connected.
\end{lemma}
\begin{proof}
  Let $Y$ be a cluster formed by using a cycle $C$ to merge clusters
  $X_1, X_2, \ldots X_q$. The edges of the cycle $C$ form a
  2-connected subgraph of $G$, and we assume that each $X_j$ is
  2-connected by induction.  Further, $C$ contains at least 2 vertices
  of each $X_j$\footnote{A cluster $X_j$ may be a singleton vertex
    (for instance, if we are in tier 0), but such a vertex does not
    affect 2-connectivity.}, so we can use induction and
  Proposition~\ref{prop:shareEdge} above: We assume $C
  \cup\{X_l\}_{l=1}^j$ is 2-connected by induction, and $C$ contains 2
  vertices of $X_{j+1}$, so $C \cup\{X_l\}_{l=1}^{j+1}$ is
  2-connected.

  Note that we have shown $Y = C \cup \{X_j\}_{j=1}^q$ is 2-connected,
  but $C$ (and hence $Y$) might contain dummy terminals and the
  corresponding dummy edges. However, each such terminal with the 2
  associated edges is a ear of $Y$; deleting them leaves $Y$
  2-connected.
\end{proof}

\begin{lemma}\label{lem:fewLeftBehind}
  The total weight of small clusters in tier $i$ that are not merged to form
  clusters of higher tiers is at most $\frac{\ell}{2 \ceil{\log k}}$.
\end{lemma}
\begin{proof}
  Assume this were not true; this means that {\sc MergeClusters} could
  find no more cycles of density at most $\alpha$ using the remaining
  small tier $i$ clusters.  But the total cost of all the edges is at
  most $cost(G)$, and the sum of terminal weights is at least
  $\frac{\ell}{2 \ceil{\log k}}$; this implies that the density of the
  graph (using the remaining terminals) is at most $2 \ceil{\log k}
  \cdot \frac{cost(G)}{\ell} = \alpha$. But by
  Theorem~\ref{thm:cycleExists}, the graph must then contain a good
  non-trivial cycle, and so the while loop would not have terminated.
\end{proof}

\begin{corollary}\label{cor:weightLargeClusters}
  When the algorithm {\sc MergeClusters} terminates, the total weight of large
  clusters is at least $\ell/2 > (4 \log k) \cdot k$.
\end{corollary}
\begin{proof}
  Each terminal not in a large cluster contributes to the weight of a
  cluster that was not merged with others to form a cluster of a
  higher tier. The previous lemma shows that the total weight of such
  clusters in any tier is at most $\frac{\ell} {2\ceil{\log k}}$;
  since there are $\ceil{\log k}$ tiers, the total number of terminals
  not in large clusters is less than $\ceil{\log k} \cdot
  \frac{\ell}{2 \ceil{\log k}} = \ell/2$.
\end{proof}

So far, we have shown that most terminals reach large clusters, all of
which are 2-connected, but we have not argued about the density of
these clusters. The next lemma says that if we can find a large
cluster of good density, we can find a solution to the \kv problem of
good density.

\begin{lemma}\label{lem:segment}
  Let $Y$ be a large cluster formed by {\sc MergeClusters}. If $Y$ has
  density at most $\delta$, we can find a graph $Y'$ with at least $k$
  terminals, each of which is $2$-connected to $r$, of total cost at
  most $2 \delta k + 2 L$.
\end{lemma}
\begin{proof}
  Let $X_1, X_2, \ldots X_q$ be the clusters merged to form $Y$ in
  order around the cycle $C$ that merged them; each $X_j$ was a small
  cluster, of weight at most $k$. A simple averaging argument shows
  that there is a consecutive segment of $X_j$s with total weight
  between $k$ and $2k$, such that the cost of the edges of $C$
  connecting these clusters, together with the costs of the clusters
  themselves, is at most $2 \delta k$. Let $X_a$ be the ``first''
  cluster of this segment, and $X_b$ the ``last''. Let $v$ and $w$ be
  arbitrary terminals of $X_a$ and $X_b$ respectively. Connect each of
  $v$ and $w$ to the root $r$ using 2 vertex-disjoint paths; the cost
  of this step is at most $2L$. (We assumed that every terminal could
  be 2-connected to $r$ using disjoint paths of cost at most $L$.)
  The graph $Y'$ thus constructed has at least $k$ terminals, and
  total cost at most $2 \delta k + 2L$.  

  We show that every vertex $z$ of $Y'$ is 2-connected to $r$; this
  completes our proof. Let $z$ be an arbitrary vertex of $Y'$; suppose
  there is a cut-vertex $x$ which, when deleted, separates $z$ from
  $r$. Both $v$ and $w$ are 2-connected to $r$, and therefore neither
  is in the same component as $z$ in $Y' - x$. However, we describe 2
  vertex-disjoint paths $P_v$ and $P_w$ in $Y'$ from $z$ to $v$ and
  $w$ respectively; deleting $x$ cannot separate $z$ from both $v$ and
  $w$, which gives a contradiction. The paths $P_v$ and $P_w$ are easy
  to find; let $X_j$ be the cluster containing $z$. The cycle $C$
  contains a path from vertex $z_1 \in X_j$ to $v' \in X_a$, and
  another (vertex-disjoint) path from $z_2 \in X_j$ to $w' \in X_b$.
  Concatenating these paths with paths from $v'$ to $v$ in $X_a$ and
  $w'$ to $w$ in $X_b$ gives us vertex-disjoint paths $P_1$ from $z_1$
  to $v$ and $P_2$ from $z_2$ to $w$. Since $X_j$ is 2-connected, we
  can find vertex-disjoint paths from $z$ to $z_1$ and $z_2$, which
  gives us the desired paths $P_v$ and $P_w$.\footnote{The vertex $z$
    may not be in any cluster $X_j$. In this case, $P_v$ is formed by
    using edges of $C$ from $z$ to $v' \in X_a$, and then a path from
    $v'$ to $v$; $P_w$ is formed similarly.}
\end{proof}

\bigskip
We now present the two analyses of density referred to earlier. The
key difference between the weaker and tighter analysis is in the way
we bound edge costs. In the former, each large cluster pays for its
edges separately, using the fact that all cycles used have density at
most $\alpha = O(\log k) \rho$. In the latter, we crucially use the fact
that small clusters which share edges are merged. Roughly speaking,
because small clusters are edge-disjoint, the average density of small
clusters must be comparable to the density of the input graph
$G$. Once an edge is in a large cluster, we can no longer use the
edge-disjointness argument. We must pay for these edges separately,
but we can bound this cost.

First, the following lemma allows us to show that every large cluster
has density at most $O(\log^2 k) \rho$.

\begin{lemma} \label{lem:tierCost}
  For any cluster $Y$ formed by {\sc MergeClusters} during iteration $i$,
  the total cost of edges in $Y$ is at most $(i+1)\cdot \alpha w_Y$.
\end{lemma}
\begin{proof}
  We prove this lemma by induction on the number of vertices in a
  cluster.  Let $\script{S}$ be the set of clusters merged using a
  cycle $C$ to form $Y$.  Let $\script{S}_1$ be the set of clusters in
  $\script{S}$ of tier $i$, and $\script{S}_2$ be $\script{S} -
  \script{S}_1$. ($\script{S}_2$ contains clusters of tiers less or
  greater than $i$ that contained an edge of $C$.)

  The cost of edges in $Y$ is at most the sum of: the cost of $C$, the
  cost of $\script{S}_1$, and the cost of $\script{S}_2$. Since all
  clusters in $\script{S}_2$ have been formed during iteration $i$ or
  earlier, and are smaller than $Y$, we can use induction to show that
  the cost of edges in $\script{S}_2$ is at most $(i+1) \alpha \sum_{X
    \in \script{S}_2}w_X$. All clusters in $\script{S}_1$ are of tier
  $i$, and so must have been formed before iteration $i$ (any cluster
  formed during iteration $i$ is of a strictly greater tier), so we
  use induction to bound the cost of edges in $\script{S}_1$ by $i
  \alpha \sum_{X \in \script{S}_1} w_X$.

  Finally, because $C$ was a good-density cycle, and only clusters of
  tier $i$ contribute to calculating the density of $C$, the cost of
  $C$ is at most $\alpha \sum_{X \in \script{S}_1} w_X$. Therefore,
  the total cost of edges in $Y$ is at most $(i+1) \alpha \sum_{X \in
    \script{S}} w_X = (i+1) \alpha w_Y$.
\end{proof}

\bigskip
Let $Y$ be an arbitrary large cluster; since we have only $\ceil{\log k}$
tiers, the previous lemma implies that the cost of $Y$ is at most
$\ceil{\log k} \cdot \alpha w_Y = O(\log^2 k) \rho w_Y$. That is, the
density of $Y$ is at most $O(\log^2 k) \rho$, and we can use this fact
together with Lemma~\ref{lem:segment} to find a solution to the rooted \kv
problem of cost at most $O(\log^2 k) \rho k + 2L$. This completes the
`weaker' analysis, but this does not suffice to prove
Theorem~\ref{thm:avekv}; to prove the theorem, we would need to use a large
cluster $Y$ of density $O(\log k) \rho$, instead of $O(\log^2 k) \rho$.

For the purpose of the more careful analysis, implicitly construct a
forest $\script{F}$ on the clusters formed by {\sc
  MergeClusters}. Initially, the vertex set of $\script{F}$ is just
$S$, the set of terminals, and $\script{F}$ has no edges. Every time a
cluster $Y$ is formed by merging $X_1, X_2, \ldots X_q$ , we add a
corresponding vertex $Y$ to the forest $\script{F}$, and add edges
from $Y$ to each of $X_1, \ldots X_q$; $Y$ is the parent of $X_1,
\ldots X_q$. We also associate a cost with each vertex in
$\script{F}$; the cost of the vertex $Y$ is the cost of the cycle used
to form $Y$ from $X_1, \ldots X_q$. We thus build up trees as the
algorithm proceeds; the root of any tree corresponds to a cluster that
has not yet become part of a bigger cluster. The leaves of the trees
correspond to vertices of $G$; they all have cost 0. Also, any large
cluster $Y$ formed by the algorithm is at the root of its tree; we
refer to this tree as $T_Y$.

For each large cluster $Y$ after {\sc MergeClusters} terminates, say
that $Y$ is of type $i$ if $Y$ was formed during iteration $i$ of
MergeClusters. We now define the \emph{final-stage} clusters of $Y$:
They are the clusters formed during iteration $i$ that became part of
$Y$. (We include $Y$ itself in the list of final-stage clusters; even
though $Y$ was formed in iteration $i$ of {\sc MergeClusters}, it may
contain other final-stage clusters. For instance, during iteration
$i$, we may merge several tier $i$ clusters to form a cluster $X$ of
tier $j > i$. Then, if we find a good-density cycle $C$ that contains
an edge of $X$, $X$ will merge with the other clusters of $C$.)  The
\emph{penultimate} clusters of $Y$ are those clusters that exist just
before the beginning of iteration $i$ and become a part of
$Y$. Equivalently, the penultimate clusters are those formed before
iteration $i$ that are the immediate children in $T_Y$ of final-stage
clusters. Figure 1 illustrates the definitions of final-stage and
penultimate clusters. Such a tree could be formed if, in iteration
$i-1$, 4 clusters of this tier merged to form $D$, a cluster of tier
$i+1$. Subsequently, in iteration $i$, clusters $H$ and $J$ merge to
form $F$. We next find a good cycle containing $E$ and $G$; $F$
contains an edge of this cycle, so these three clusters are merged to
form $B$. Note that the cost of this cycle is paid for the by the
weights of $E$ and $G$ only; $F$ is a tier $i+1$ cluster, and so its
weight is not included in the density calculation. Finally, we find a
good cycle paid for by $A$ and $C$; since $B$ and $D$ share edges with
this cycle, they all merge to form the large cluster $Y$.

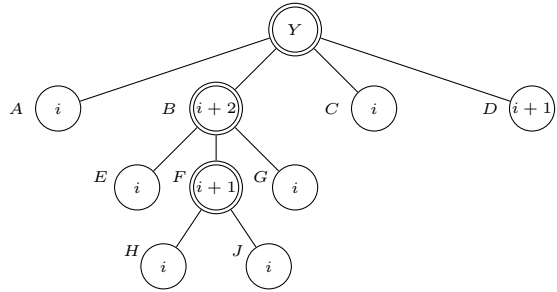
\begin{figure}[tbh]
  \begin{center}
  \begin{tikzpicture}[scale=0.7]
    \tikzstyle{vertex}=[circle,draw,inner sep=0pt,minimum size=6mm];
    \tikzstyle{high}=[circle,draw,inner sep=0pt,minimum size=7mm];

    \tikzstyle{every node}=[font=\tiny];

    \node (y) at (6,5.5) [high] {$Y$};

    \node (a) at (1.5,4) [vertex] {$i$}; \node (b) at (4.5,4) [high] {$i+2$};
    \node (c) at (7.5,4) [vertex] {$i$}; \node (d) at (10.5,4) [vertex] {$i+1$}; 
    \draw (a) -- (y) -- (b); \draw (c) -- (y) -- (d);

    \node at (0.7,4) {$A$}; \node at (3.6,4) {$B$}; \node at (6.7,4) {$C$};
    \node at (9.7,4) {$D$};


    \node (e) at (3,2.5) [vertex] {$i$}; \node (f) at (4.5,2.5) [high] {$i+1$};
    \node (g) at (6,2.5) [vertex] {$i$};
    \draw (e) -- (b) -- (f); \draw (b) -- (g);

    \node at (2.3,2.7) {$E$}; \node at (3.8,2.7) {$F$}; \node at (5.35,2.7) {$G$};

    \node (h) at (3.5,1) [vertex] {$i$}; \node (j) at (5.5,1) [vertex] {$i$};
    \draw (h) -- (f) -- (j); 

    \node at (2.9,1.3) {$H$}; \node at (4.9,1.3) {$J$}; 





    \node at (y) [vertex] {}; \node at (b) [vertex] {}; \node at (f) [vertex] {};

  \end{tikzpicture}
  \end{center}
  \caption{A part of the Tree $T_Y$ corresponding to $Y$, a large
    cluster of type $i$. The number in each vertex indicates the tier
    of the corresponding cluster.  Only final-stage and penultimate
    clusters are shown: final-stage clusters are indicated with a
    double circle; all other clusters are penultimate.}
\end{figure}

An edge of a large cluster $Y$ is said to be a \emph{final edge} if it
is used in a cycle $C$ that produces a final-stage cluster of $Y$. All
other edges of $Y$ are called \emph{penultimate edges}; note that any
penultimate edge is in some penultimate cluster of $Y$. We define the
\emph{final cost} of $Y$ to be the sum of the costs of its final
edges, and its \emph{penultimate cost} to be the sum of the costs of
its penultimate edges; clearly, the cost of $Y$ is the sum of its
final and penultimate costs. We bound the final costs and penultimate
costs separately.

Recall that an edge is a final edge of a large cluster $Y$ if it is
used by {\sc MergeClusters} to form a cycle $C$ in the final iteration
during which $Y$ is formed. The reason we can bound the cost of final
edges is that the cost of any such cycle is at most $\alpha$ times the
weight of clusters contained in the cycle, and a cluster does not
contribute to the weight of more than one cycle in an iteration. (This
is also the essence of Lemma~\ref{lem:tierCost}.) We formalize this
intuition in the next lemma.

\begin{lemma}\label{lem:final}
  The final cost of any large cluster $Y$ is at most $\alpha w_Y$,
  where $w_Y$ is the weight of $Y$.
\end{lemma}
\begin{proof}
  Let $Y$ be an arbitrary large cluster. In the construction of the
  tree $T_Y$, we associated with each vertex of $T_Y$ the cost of the
  cycle used to form the corresponding cluster. To bound the total
  final cost of $Y$, we must bound the sum of the costs of vertices of
  $T_Y$ associated with final-stage clusters.  The weight of $Y$,
  $w_Y$ is at least the sum of the weights of the penultimate tier $i$
  clusters that become a part of $Y$. Therefore, it suffices to show
  that the sum of the costs of vertices of $T_Y$ associated with
  final-stage clusters is at most $\alpha$ times the sum of the
  weights of $Y$'s penultimate tier $i$ clusters. (Note that a tier
  $i$ cluster must have been formed prior to iteration $i$, and hence
  it cannot itself be a final-stage cluster.)

  A cycle was used to construct a final-stage cluster $X$ only if its
  cost was at most $\alpha$ times the sum of weights of the
  penultimate tier $i$ clusters that become a part of $X$. (Larger
  clusters may become a part of $X$, but they do not contribute weight
  to the density calculation.)  Therefore, if $X$ is a vertex of $T_Y$
  corresponding to a final-stage cluster, the cost of $X$ is at most
  $\alpha$ times the sum of the weights of its tier $i$ immediate
  children in $T_Y$. But $T_Y$ is a tree, and so no vertex
  corresponding to an penultimate tier $i$ cluster has more than one
  parent. That is, the weight of a penultimate cluster pays for only
  one final-stage cluster. Therefore, the sum of the costs of vertices
  associated with final-stage clusters is at most $\alpha$ times the
  sum of the weights of $Y$'s penultimate tier $i$ clusters, and so
  the final cost of $Y$ is at most $\alpha w_Y$.
\end{proof}

\begin{lemma}\label{lem:penultimate}
  If $Y_1$ and $Y_2$ are distinct large clusters of the same type, no
  edge is a penultimate edge of both $Y_1$ and $Y_2$.
\end{lemma}
\begin{proof}
  Suppose, by way of contradiction, that some edge $e$ is a
  penultimate edge of both $Y_1$ and $Y_2$, which are large clusters
  of type $i$. Let $X_1$ (respectively $X_2$) be a penultimate cluster
  of $Y_1$ (resp. $Y_2$) containing $e$. As penultimate clusters, both
  $X_1$ and $X_2$ are formed before iteration $i$. But until iteration
  $i$, neither is part of a large cluster, and two small clusters
  cannot share an edge without being merged. Therefore, $X_1$ and
  $X_2$ must have been merged, so they cannot belong to distinct large
  clusters, giving the desired contradiction.
\end{proof}

\begin{theorem}\label{thm:goodLargeCluster}
  After {\sc MergeClusters} terminates, at least one large cluster has
  density at most $O(\log k) \rho$.
\end{theorem}
\begin{proof}
  We define the \emph{penultimate density} of a large cluster to be
  the ratio of its penultimate cost to its weight.

  Consider the total penultimate costs of all large clusters: For any
  $i$, each edge $e \in E(G)$ can be a penultimate edge of at most 1
  large cluster of type $i$. This implies that each edge can be a
  penultimate edge of at most $\ceil{\log k}$ clusters. Therefore, the
  sum of penultimate costs of all large clusters is at most
  $\ceil{\log k} cost(G)$. Further, the total weight of all large
  clusters is at least $\ell/2$.  Therefore, the (weighted) average
  penultimate density of large clusters is at most $2 \ceil{\log k}
  \frac{cost(G)}{\ell} = 2 \ceil{\log k} \rho$, and hence there exists a
  large cluster $Y$ of penultimate density at most $2 \ceil{\log k}
  \rho$.

  The penultimate cost of $Y$ is, therefore, at most $2 \ceil{\log k}
  \rho w_Y$, and from Lemma~\ref{lem:final}, the final cost of $Y$
  is at most $\alpha w_Y$.  Therefore, the density of $Y$ is at most
  $\alpha + 2 \ceil{\log k} \rho = O(\log k) \rho$.
\end{proof}

Theorem~\ref{thm:goodLargeCluster} and Lemma \ref{lem:segment}
together imply that we can find a solution to the rooted \kv problem
of cost at most $O(\log k) \rho k + 2L$. This completes our proof of
Theorem~\ref{thm:avekv}.

\section{Conclusions}
\label{sec:conclusion}

We list the following open problems:
\begin{itemize}
\item Can the approximation ratio for the \kv problem be improved from
  the current $O(\log \ell \log k)$ to $O(\log n)$ or better? Removing
  the dependence on $\ell$ to obtain even $O(\log^2 k)$ could be
  interesting. If not, can one improve the approximation ratio for the
  easier \ke problem?

\item Can we obtain approximation algorithms for the \kvc{\lambda} or
  \kec{\lambda} problems for $\lambda > 2$? In general, few results
  are known for problems where vertex-connectivity is required to be
  greater than 2, but there has been more progress with higher
  edge-connectivity requirements.

\item Given a 2-connected graph of density $\rho$ with some vertices
  marked as terminals, we show that it contains a non-trivial cycle
  with density at most $\rho$, and give an algorithm to find such a
  cycle. We have also found an $O(\log \ell)$-approximation for the
  problem of finding a minimum-density non-trivial cycle. Is there a
  constant-factor approximation for this problem? Can it be solved
  \emph{exactly} in polynomial time?
\end{itemize}

\medskip
\noindent
\textbf{Acknowledgments:} We thank Mohammad Salavatipour for helpful
discussions on \ke and related problems. We thank Erin Wolf Chambers
for useful suggestions on notation.

\bibliographystyle{plain}

\end{document}